\documentclass[10pt]{birkjour}
 \newtheorem{thm}{Theorem}
 \newtheorem{cor}[thm]{Corollary}

 \theoremstyle{definition}
 
 \theoremstyle{remark}

 \numberwithin{equation}{section}
 \usepackage{amsmath}
 \usepackage{amssymb,textcomp}
\author{Helmut Kahl}
\address{
Munich University of Applied Sciences\\
Lothstr. 34\\
D-80335 M\"unchen\\
Germany}
\email{kahl@hm.edu}
\title{The loss value of multilinear regression}

\begin{document}
\maketitle

{\bf Abstract.}	Determinant formulas are presented for:
	a certain positive semidefinite, hermitian matrix;
	the loss value of multilinear regression;
	the multiple linear regression coefficient.

{\bf Mathematical Subject Cassification (2010).} 15A03, 15A15, 62H12.


\section{Euclidean distance by help of determinants}

\paragraph{Introduction} We consider the euclidean norm $x \mapsto \|x\| := \sqrt{x^* x}$ of column vectors $x \in \mathbb{C}^m$. For a complex $m \times n$-matrix $A \in \mathbb{C}^{m \times n}$ and a vector $b \in \mathbb{C}^m$ we denote by $(A|b)$ the $m \times (n+1)$-matrix $A$ together with $b$ as the last column. The following theorem is well-known in the real case and can be shown by help of a formula for the volume of an $n$-dimensional parallelepiped embedded in $\mathbb{R}^m$. Here we offer a less known proof via QR-decomposition of a complex matrix.

\begin{thm}\label{thm}
	For the euclidean distance $\textnormal{dist}(A, b)$ between $b \in \mathbb{C}^m$ and the column space of $A \in \mathbb{C}^{m \times n}$ it holds:
	\begin{equation*}
		\textnormal{dist}(A, b) \sqrt{\det(A^* A)} = \sqrt{\det((A|b)^* (A|b))}
	\end{equation*}
\end{thm}
\begin{proof}
	For a unitary $m \times m$-matrix $Q$ it holds $\textnormal{dist}(A, b)$ = $\textnormal{dist}(Q A, Q b)$ due to \cite{Horn/Johnson}, thm. 2.1.4(g). According to \cite{Horn/Johnson}, thm. 2.1.14(d) there is a unitary $Q$ s.t. $(A|b) = Q (R|c)$ with $(R|c)$ upper right triangular. Hence $\delta = \textnormal{dist}(R, c)$ is the absolute value of the $(n+1)$-th coordinate of $c$. So we have
	\begin{equation*}
		\det((A|b)^* (A|b)) = \det((R|c)^* (R|c)) = \delta^2 \det(R^* R) = \delta^2 \det(A^* A) .
	\end{equation*}
	Since $\det(A^* A) \ge 0$ for arbitrary complex matrices $A$, s. e.g. \cite{Horn/Johnson}, thm. 4.1.5 \& 7.2.7(a), the assertion follows.
\end{proof}

\paragraph{Matrix equation for the distance} In case $A$ has full column rank $\textnormal{rk}(A) = n$ we have $\det(A^* A) > 0$. Then the formula yields $\textnormal{dist}(A, b)$ as a quotient of the two square root values. And by plugging the minimum point $x = (A^* A)^{-1} A^* b$ into $\|A x - b\|^2$ we obtain
\begin{equation*}
	\det((A|b)^* (A|b)) / \det(A^* A) = \textnormal{dist}(A, b)^2 = b^* b - b^* A (A^* A)^{-1} A^* b .
\end{equation*}

\paragraph{Special determinant equation} Now, for a matrix $A \in \mathbb{C}^{(n+1) \times n}$ let $A_i$ denote the matrix $A$ without its $i$-th row. Via developing $\det(A|a_j) = 0$ by the last column for every column $a_j$ of $A$ we see that the vector
\begin{equation*}
	b := \left((-1)^n \overline{\det(A_i)}\right)_{i = 1 , ... ,n+1}	
\end{equation*}
is orthogonal to the column space of $A$. Applying Theorem \ref{thm} to $A$ and $b$ we obtain

\begin{cor}
	For $A \in \mathbb{C}^{(n+1) \times n}$ holds the identity of $n \times n$-determinants:
	\begin{equation*}
		\sum_{i=1}^{n+1} |\det(A_i)|^2 = \det(A^* A)
	\end{equation*}
\end{cor}

\section{Loss value and correlation}

\paragraph{Introduction} The task of \textit{multiple linear regression} is the computation of \textit{regression coefficients} $\alpha_0, \alpha_1 , ... , \alpha_n$ of the \textit{fitting hyperplane} (in $\mathbb{R}^{n+1}$)
\begin{equation*}
	y = \alpha_0 + \alpha_1 x_1 + ... + \alpha_n x_n	
\end{equation*}
as a function of variables $x_1 , ... , x_n \in \mathbb{R}$ from (\textit{empirical}) data points
\begin{equation*}
	(x_{1 1} , ... , x_{1 n} , y_1) , ... , (x_{m 1} , ... , x_{m n} , y_m) \in \mathbb{R}^{n+1} , m \in \mathbb{N}
\end{equation*}
s.t. the \textit{loss value}
\begin{equation*}
	\delta := \left(\sum_{i=1}^{m} (\alpha_0 + \alpha_1 x_{i 1} + ... + \alpha_n x_{i n} - y_i)^2\right)^{1/2}
\end{equation*}
is at minimum. For $a := (\alpha_0, \alpha_1 , ... , \alpha_n)^t$, $y := (y_1 , ... , y_m)^t$ and the matrix $(1|X)$ that we obtain from $X := (x_{i j})_{i \in \mathbb{N}_m, j \in \mathbb{N}_n}$ by prepending $(1 , ... , 1)^t \in \mathbb{R}^m$ as an extra column (of index $0$) we have $\delta = \|(1|X) a - y\|$. So the minimal value of $\delta$ is the euclidean distance between $y$ and the column space of $(1|X)$.

\paragraph{Centering} In statistics it is common to express empirical values of expectation with help of the arithmetic mean $\bar{y} := (y_1 + ... + y_m)/m$ of a (\textit{sample}) vector like $y$ above. A regression vector $a$ like described above is defined by the normal equation system
\begin{equation}\label{eq_normal}
	(1|X)^t (1|X) a = (1|X)^t y .	
\end{equation}
After division by $m$ the equation of row index $0$ of equation \ref{eq_normal} ends in
\begin{equation}\label{eq_mean}
	\bar y = \alpha_0 + \alpha_1 \bar{x}_1 + ... + \alpha_n \bar{x}_n	
\end{equation}
where $x_j$ denotes the $j$-th column of $X$. We denote by $\hat{y} := (y_1 - \bar{y}, ... , y_m - \bar{y})^t$ the \textit{centering of} $y$ and by $\hat{X}$ the $m \times n$-matrix obtained from $X$ by centering all its columns.\footnote{Then $(\hat{X}^t \hat{X})/(m-1)$ is the \textit{sample covariance matrix of the sample matrix} $X$. It serves as an estimator of the \textit{covariance matrix of the random vector} $(X_1 , ... , X_n)$ whose $m$ samples are given by $X$, row by row. With the additional random variable $Y$ whose samples are represented by $y$ the \textit{mean squared loss value} $\delta^2/(m-1)$ is an estimator of the expected value of the random variable $(Y - \alpha_0 - \alpha_1 X_1 - ... - \alpha_n X_n)^2$; s. e.g. \cite{Beichelt/Montgomery}, Kap. 3.8!} Then the normal equations of row indices $1$ to $n$ of equation \ref{eq_normal} are transformed to
\begin{equation}\label{eq_center}
	\hat{X}^t \hat{X} a_1 = \hat{X}^t \hat{y} , a_1 := (\alpha_1 , ... , \alpha_n)^t
\end{equation}
by subtracting the $\bar{x}_i$-th multiple of equation \ref{eq_mean} from the $i$-th normal equation for $i = 1 , ... , n$. This shows $\textnormal{rk}(1|X) = \textnormal{rk}(\hat{X}) + 1$.

\begin{thm}\label{thm_lossval}
	In case $\textnormal{rk}(1|X) = n+1$ the loss value of the sample matrix $(X|y)$ equals
	\begin{equation*}
		\sqrt{\det\left(\left(\hat{X}|\hat{y}\right)^t \left(\hat{X}|\hat{y}\right)\right) / \det\left((\hat{X})^t \hat{X}\right)}
	\end{equation*}
\end{thm}
\begin{proof}
	Expressing $\alpha_0$ in terms of the other regression coefficients by help of equation \ref{eq_mean} gives us $y - \alpha_0 - \alpha_1 x_1 - ... - \alpha_n x_n = \hat{y} - \alpha_1 \hat{x}_1 - ... - \alpha_n \hat{x}_n$. Hence the loss value is the euclidean distance between $\hat{y}$ and the column space of $\hat{X}$. Because $\hat{X}$ has full rank the formula follows by Theorem \ref{thm}.
\end{proof}

\paragraph{Correlation} For the orthogonal projection $p := (1|X) a$ of $y$ onto the column space of $(1|X)$ it holds $\hat{p} = \hat{X} a_1$. So by equation \ref{eq_center} $\hat{p}$ is the orthogonal projection of $\hat{y}$ onto the column space of $\hat{X}$. Hence in case $\hat{y}, \hat{p} \ne 0$ the angle between $\hat{y}$ and $\hat{p}$ is at most $\pi / 2$.\footnote{The condition $\hat{y} \ne 0$ means a non-zero \textit{sample variance} $(\hat{y}^t \hat{y})/(m-1)$ of $y$.} Therefore the \textit{multiple correlation coefficient}
\begin{equation*}
	\rho(X,y) := \hat{y}^t \hat{p} / \|\hat{p}\| / \|\hat{y}\|	
\end{equation*}
between $y$ and $X$ is non-negative. According to the Cauchy-Schwarz inequality it is at most $1$. The latter theorem allows the computation of $\rho(X,y)$ without the computation of $p$, i.e. without performing the linear regression.

\begin{cor}
	For a sample vector $y \in \mathbb{R}^m$ with $\hat{y} \ne 0$ and a sample matrix $X \in \mathbb{R}^{m \times n}$ with $\textnormal{rk}(\hat{X}) = n$ it holds
	\begin{equation*}
		\rho(X,y) = \sqrt{1 - \det\left(\left(\hat{X}|\hat{y}\right)^t \left(\hat{X}|\hat{y}\right)\right) / \left(\det\left((\hat{X})^t \hat{X}\right) \hat{y}^t \hat{y}\right)} 
	\end{equation*}
\end{cor}
\begin{proof}
	The assertion follows from Theorem \ref{thm_lossval} by the Theorem of Pythagoras applied to $\hat{y}/ \|\hat{y}\|$ as the hypotenuse and $\hat{p} / \|\hat{y}\|$ as a cathetus.
\end{proof}

\end{document}